\documentclass[a4paper, 10pt, twosides]{amsart}

\usepackage{amssymb}
\usepackage{enumitem}
\usepackage{dsfont}
\usepackage[usenames, dvipsnames]{color}
\usepackage{tikz}
\usepackage{subfig}
\usepackage[colorlinks=true,
		raiselinks=true,
		linkcolor=MidnightBlue,
		citecolor=Mahogany,
		urlcolor=ForestGreen,
		pdfauthor={Atreyee Kundu},
		pdftitle={Stabilizing switching signals for switched linear systems},
		pdfkeywords={switched systems, asymptotic stability},
		pdfsubject={Technical Report},
		plainpages=false]{hyperref}

\usepackage{mathtools}
\allowdisplaybreaks
\newtheorem{theorem}{Theorem}[section]
\newtheorem{prop}[theorem]{Proposition}

\theoremstyle{definition}
\newtheorem{definition}[theorem]{Definition}

\theoremstyle{remark}
\newtheorem{remark}[theorem]{Remark}

\theoremstyle{assumption}

\theoremstyle{fact}
\newtheorem{fact}[theorem]{Fact}

\theoremstyle{claim}

\theoremstyle{cor}

\numberwithin{equation}{section}

\newcommand{\norm}[1]{\left\lVert{#1}\right\rVert} 
\newcommand{\abs}[1]{\left\lvert{#1}\right\rvert} 

\newcommand{\R}{\mathbb{R}} 
\renewcommand{\P}{\mathcal{P}}  
\newcommand{\ls}{\limsup}   
\newcommand{\li}{\liminf}   
\newcommand{\ol}{\overline}

\newcommand{\Ntsigma}{N_{t}^{\sigma}} 
\newcommand{\swinstants}{0 \teL \tau_{0}<\tau_{1}<\tau_{2}<\cdots<\tau_{\Ntsigma}} 

\newcommand{\muprod}{\prod_{i=0}^{\Ntsigma}\mu_{\sigma(\tau_{i})\sigma(\tau_{i+1})}} 
\newcommand{\dift}{t-\tau_{\Ntsigma}} 
\newcommand{\diftau}{\tau_{i+1}-\tau_{i}} 
\newcommand{\tausum}{\sum_{i=0}^{\Ntsigma-1}\lambda_{\sigma(\tau_{i})}(\diftau)} 
\newcommand{\ktol}{\#{\{k\rightarrow\ell\}_{\Ntsigma}}} 
\newcommand{\indicator}{\mathds{1}_{\{j\}}(\sigma(\tau_{i}))} 
\newcommand{\holdt}{S_{i+1}}   
\newcommand{\nuh}{\frac{\Ntsigma}{h(t)}}    

\newcommand{\lra}{\longrightarrow}
\newcommand{\Let}{\coloneqq}
\newcommand{\teL}{\eqqcolon}

\newcommand{\eps}{\varepsilon}
\newcommand{\drv}{\mathrm d}
\newcommand{\transp}{^{\top}}
\newcommand{\pmat}[1]{\begin{pmatrix}#1\end{pmatrix}}
\newcommand{\inprod}[2]{\left\langle #1 , #2 \right\rangle}
\renewcommand{\geq}{\geqslant}
\renewcommand{\ge}{\geqslant}
\renewcommand{\leq}{\leqslant}
\renewcommand{\le}{\leqslant}

\newcommand{\astablesum}{\sum_{j\in \P_{AS}}\lambda_{j}\sum_{i:\sigma(\tau_{i})=j}\holdt} 
\newcommand{\mstablesum}{\sum_{j\in \P_{MS}}\lambda_{j}\sum_{i:\sigma(\tau_{i})=j}\holdt} 
\newcommand{\ustablesum}{\sum_{j\in \P_{U}}\lambda_{j}\sum_{i:\sigma(\tau_{i})=j}\holdt}  
\newcommand{\astablesummag}{\sum_{j\in \P_{AS}}\abs{\lambda_{j}}\sum_{i:\sigma(\tau_{i})=j}\holdt} 
\newcommand{\ustablesummag}{\sum_{j\in \P_{U}}\abs{\lambda_{j}}\sum_{i:\sigma(\tau_{i})=j}\holdt} 
\newcommand{\ustablesummagfreq}{\sum_{j\in \P_{U}}\abs{\lambda_{j}}\sum_{i:\sigma(\tau_{i})=j}\frac{\holdt}{h(t)}} 
\newcommand{\astablesummagfreq}{\sum_{j\in \P_{AS}}\abs{\lambda_{j}}\sum_{i:\sigma(\tau_{i})=j}\frac{\holdt}{h(t)}} 

\newcommand{\ustablesummagfreqnew}{\sum_{j\in \P_{U}}\abs{\lambda_{j}}\eta_{h}(j,t)} 
\newcommand{\astablesummagfreqnew}{\sum_{j\in \P_{AS}}\abs{\lambda_{j}}\eta_{h}(j,t)} 

\newcommand{\ustablesummagfreqls}{\sum_{j\in \P_{U}}\abs{\lambda_{j}}\ls_{t\rightarrow+\infty}\eta_{h}(j,t)} 
\newcommand{\astablesummagfreqli}{\sum_{j\in \P_{AS}}\abs{\lambda_{j}}\li_{t\rightarrow+\infty}\eta_{h}(j,t)} 

\title[Stabilizing switching signals for switched linear systems]{Stabilizing switching signals for\\switched linear systems}
\thanks{The authors thank Daniel Liberzon for helpful discussions and pointers to several relevant articles.}
\author[A.\ Kundu]{Atreyee Kundu}
\address{Systems \& Control Engineering, IIT Bombay, Mumbai~--~400076, India}
\email[A.\ Kundu]{atreyee@sc.iitb.ac.in}
\author[D.\ Chatterjee]{Debasish Chatterjee}
\email[D.\ Chatterjee]{dchatter@iitb.ac.in}
\urladdr[D.\ Chatterjee]{\url{http://www.sc.iitb.ac.in/~chatterjee}}

\date{\today}

\begin{document}
    \maketitle

    \begin{abstract}
        This article deals with stability of continuous-time switched linear systems under constrained switching. Given a family of linear systems, possibly containing unstable dynamics, we characterize a new class of switching signals under which the switched linear system generated by it and the family of systems is globally asymptotically stable. Our characterization of such stabilizing switching signals involves the asymptotic frequency of switching, the asymptotic fraction of activation of the constituent systems, and the asymptotic densities of admissible transitions among them. Our techniques employ multiple Lyapunov-like functions, and extend preceding results both in scope and applicability.
    \end{abstract}

    \section{Introduction}
	\label{s:intro}
        A \emph{switched system} consists of a family of systems, and a \emph{switching signal} that selects an \emph{active system} from the family at every instant of time \cite[\S1.1.2]{Liberzon}. Switched systems have been employed to model dynamical behavior that are subject to known or unknown abrupt parameter variations \cite{Sun}, and naturally arise in a multitude of application areas such as networked systems, quantization, variable structure systems, etc; see e.g., \cite{Heemels_survey, Heemels_switched, Heemels_lecture, Shorten_review, Antsaklis_survey} and the references therein.

		This article is concerned with stability of switched linear systems. Stability of switched systems has been broadly classified into two categories---\emph{stability under arbitrary switching}, which is concerned with finding conditions that guarantee asymptotic stability of a switched system under all possible switching signals \cite[Chapter 2]{Liberzon}, and \emph{stability under constrained switching}, which is concerned with identifying classes of switching signals for which the switched system is asymptotically stable \cite[Chapter 3]{Liberzon}. Here we concentrate on stability under constrained switching.

        Research in the direction of stability under constrained switching has primarily focussed on the concept of slow switching among asymptotically stable systems. The key idea behind stability under \emph{slow switching} \cite[$\S$3.2]{Liberzon} is that if all the constituent systems are stable and the switching is sufficiently slow, then the ``energy injected due to switching'' gets sufficient time for ``dissipation'' due to the asymptotic stability of the individual systems. This idea is captured well by the concepts of dwell time \cite[$\S$3.2.1]{Liberzon} and average dwell time \cite[$\S$3.2.2]{Liberzon} switching. For instance, a switched linear system is asymptotically stable if all the systems in the family are asymptotically stable, and the dwell time of the switching signal is sufficiently large \cite{Antsaklis_survey}. A similar assertion holds under the more general class of average dwell time switching signals \cite{HespanhaMorse, Liberzon}, where the number of switches on any time interval can grow at most as an affine function of the length of the interval \cite[Chapter 3]{Liberzon}. Stability analysis of switched systems under these switching signals is now considered classical \cite{Antsaklis_survey, Shorten_review}. Several other classes of stabilizing switching signals, e.g., those obeying weak and persistent dwell time conditions have been proposed in \cite{Hespanha_Lasalles, Goebel}. Asymptotic properties of switched systems are discussed in \cite{linear_convergence} under weak, persistent, and strong dwell time defined on each constituent system and not as a property of the switching signal, unlike their earlier counterparts. In the presence of unstable systems in the family, however, most of the above-mentioned results do not carry over in a straightforward fashion. Indeed, slow switching alone cannot guarantee stability of the switched system---additional conditions are essential to ensure that the switched system does not spend too much time in the unstable components \cite{Antsaklis_survey}. It is thus no longer sufficient to characterize switching signals with the frequency of switching.

        In the light of the aforementioned works, our contributions are as follows:
		\begin{itemize}[label=\(\circ\), leftmargin=*]
			\item We admit switching signals having properties beyond the average dwell time regime. Our switching signals can have a frequency of switching that grows, for example, linearly with time; in contrast, average dwell time switching restricts the switching frequency to be \(O(1)\).
			\item We characterize stabilizing switching signals by explicitly involving the following asymptotic properties of the switching signal: the frequency of switching, the fraction of activity of the constituent systems, and the ``density'' of the admissible transitions among them. Our characterizations do not involve point-wise bounds on the number of switches. This is in contrast to average dwell time switching, where the characterization requires a uniform bound on the number of switches on \emph{every} interval of time of equal length.
			\item We address global asymptotic stability of a switched linear system containing unstable components. While this is not the first time that stability of switched systems containing unstable dynamics in the family has been considered (see e.g., \cite{Zhai01, Liberzon_IOSS},) we contend that this is the first instance where stability of a switched linear system is treated in the presence of unstable dynamics in the family, \emph{and} the conditions for stability involve only asymptotic properties of the switching signal.
		\end{itemize}

         The article unfolds as follows: In \S\ref{s:prelims} we formalize the setting for our result. Here we introduce our chief analytical apparatus---multiple Lyapunov-like functions \cite[\S3.1]{Liberzon} and a directed graph on the set of constituent systems induced by the switching signal to represent the set of admissible transitions. We identify and derive key properties of the family of systems and asymptotic properties of the switching signals that we need. Our main result is stated in \S\ref{s:mainres}, where we elaborate, in a sequence of remarks, on its various aspects and its relationship to preceding works. We illustrate our main result with the aid of a numerical example in \S\ref{s:example}, and conclude in \S\ref{s:conclusion} with a brief summary of future directions. The proofs of all the claims are collected in Appendices \ref{ap1} and \ref{ap2}.

    \section{Preliminaries}
	\label{s:prelims}
        We consider a family of continuous-time linear autonomous dynamical systems
        \begin{equation}
        \label{e:family}
            \dot{x}(t) = A_{i}x(t),\quad A_{i}\in\R^{{d}\times{d}},\quad i\in\P, \quad t\geq0,
        \end{equation}
        where $x(t)\in\R^{d}$ is the vector of states at time $t$, and $\P = \{1,2,...,N\}$ denotes a finite index set. We assume that for each $i\in\P$, the matrix $A_{i}$ has full rank; consequently, $0\in\R^{d}$ is the unique equilibrium point for each system in \eqref{e:family}. Let $\sigma:[0,+\infty[\:\lra\P$ be a piecewise constant function that specifies, at each time $t$, the system $A_{\sigma(t)}$ from the family \eqref{e:family}, that is active, at $t$. This function $\sigma$ is called the \emph{switching signal}, and by convention, $\sigma$ is assumed to be continuous from right and having limits from the left everywhere.
        The switched linear system \cite{Liberzon} generated by $\sigma$ and the family of systems \eqref{e:family} is
        \begin{equation}
        \label{e:swsys}
            \dot{x}(t) = A_{\sigma(t)}x(t),\quad x(0) = x_{0} \text{(given)}, \quad t\geq0.
        \end{equation}
        Let $0\teL\tau_{0}<\tau_{1}<\tau_{2}<\cdots$ denote the points of discontinuity of $\sigma$; these are the switching instants. For $t>0$, let $\Ntsigma$ denote the number of switches on $]0,t]$. For us a switching signal is admissible if it is piecewise constant as a function from $[0,+\infty[$ into $\P$. The solution $(x(t))_{t\geq0}$ to the switched system \eqref{e:swsys} corresponding to an admissible switching signal $\sigma$ is the map $x:[0,+\infty[\:\lra\R^{d}$ defined by
        \begin{align*}
       \label{e:swsol}
            x(t) = &\exp({A_{\sigma(\tau_{\Ntsigma})}(t-\tau_{\Ntsigma})})\exp({A_{\sigma(\tau_{{\Ntsigma}-1})}(\tau_{\Ntsigma}-\tau_{{\Ntsigma}-1})})\cdots\nonumber\\ & \exp({A_\sigma({\tau_1})(\tau_2-\tau_1)})\exp({A_\sigma({\tau_0})(\tau_1-\tau_0)})x_{0}, \quad t\geq0,
        \end{align*}
        where we have suppressed the dependence of $x$ on $\sigma$.

        Let the admissible transitions among the systems in the family \eqref{e:family} be generated by walks on a directed graph $(\P, E(\P))$.\footnote{A directed graph representation of a switching signal has appeared before in, e.g., \cite{Mancilladigraph}.} Recall that a directed graph is a set of nodes connected by edges, where each edge has a direction associated to it. The set of vertices of the graph is the set of indices of the systems, i.e., $\P$; the transitions from one system to another that are effected by walks $\sigma$ at the switching instants $\tau_{0},\tau_{1},\cdots$ are represented by directed edges among the vertices in $\P$. For $t>0$ let $\sigma|_{[0,t]}$ denote the restriction of $\sigma$ to $[0,t]$. The set of edges $E(\P)$ consists of the transitions $(\sigma(\tau_{0}),\sigma(\tau_{1})),\cdots,(\sigma(\tau_{\Ntsigma-1}),\sigma(\tau_{\Ntsigma}))$ corresponding to $\sigma|_{[0,t]}$. Note that for $\sigma|_{[0,t]}$, we have a finite number of transitions $\sigma(\tau_{0}) \rightarrow \sigma(\tau_{1}) \rightarrow \cdots \rightarrow \sigma(\tau_{\Ntsigma})$ due to the assumption that $\sigma$ is a piecewise constant function. For instance, let $\P = \{1,2,3,4\}$. Let $\sigma$ be an admissible switching signal such that the transitions $1\rightarrow 2$, $1\rightarrow 3$, $2\rightarrow 3$, and $3\rightarrow 1$ are admissible. The corresponding digraph representation is as follows:

        \begin{center}
        \begin{tikzpicture}[every path/.style={>=latex},every node/.style={draw,circle}]
            \node            (a) at (0,0)  { 1 };
            \node            (b) at (2,0)  { 2 };
            \node            (c) at (2,-2) { 3 };
            \node            (d) at (0,-2) { 4 };

            \draw[->] (a) edge (b);
            \draw[->] (b) edge (c);
            \draw[->] (a) edge (c);
            \draw[<-] (a) edge [bend right] (c);
        \end{tikzpicture}
        \end{center}

        We focus on global asymptotic stability of \eqref{e:swsys}, defined as:
        \begin{definition} \label{d:GAS}
            The switched system \eqref{e:swsys} is \emph{globally asymptotically stable} (GAS) for a given switching signal $\sigma$ if \eqref{e:swsys} is
            \begin{itemize}[label=\(\circ\), leftmargin=*]
                \item Lyapunov stable, and
                \item uniformly globally asymptotically convergent, i.e., for all \(r, \eps > 0\) there exists \(T(r, \eps) > 0\) such that \(\norm{x(t)} < \eps\) for all \(t > T(r, \eps)\) whenever $\norm{x_{0}} < r$.
            \end{itemize}
        \end{definition}
		In other words, \eqref{e:swsys} is GAS for a given switching signal \(\sigma\) if there exists a class-\(\mathcal KL\) function \(\beta_\sigma\) such that \(\norm{x(t)} \le \beta_\sigma(\norm{x_0}, t)\) for all \(x_0\in\R^d\) and \(t\geq 0\).

        \begin{fact}
        \label{fact:gas}
            For a given switching signal \(\sigma\), global asymptotic convergence of \eqref{e:swsys} is sufficient for GAS of \eqref{e:swsys}. Consequently, Lyapunov stability is not essential as a part of the definition of GAS in Definition \ref{d:GAS}.
        \end{fact}
        We provide a short direct proof of Fact \ref{fact:gas} in Appendix \ref{ap1}.

		\subsection{Properties of the family \eqref{e:family}}
        Our study on GAS of the switched system \eqref{e:swsys} requires analyzing the temporal behaviour of Lyapunov-like functions for the individual systems in the family \eqref{e:family} along the corresponding system trajectories. To this end, we need to capture quantitative measures of (in)stability of the systems in \eqref{e:family}, and our primary tool is:
        \begin{fact}
        \label{fact:key}
            For each $i\in\P$, there exists a pair $(P_{i},\lambda_{i})$, where $P_{i}\in\R^{d\times d}$ is a symmetric and positive definite matrix, and
        \begin{itemize}[label=\(\circ\), leftmargin=*]
           \item if $A_{i}$ is asymptotically stable, then $\lambda_{i}>0$;
           \item if $A_{i}$ is marginally stable, then $\lambda_{i}=0$; \footnote{Marginal stable systems are those systems in \eqref{e:family} that are Lyapunov stable but not asymptotically stable.}
           \item if $A_{i}$ is unstable, then $\lambda_{i} < 0$;
        \end{itemize}
        such that, with
        \begin{equation}
        \label{e:Lyaplike}
            \R^{d}\ni\xi\longmapsto V_{i}(\xi)\Let\langle P_{i}\xi, \xi\rangle\in[0,+\infty[,
        \end{equation}
        we have
        \begin{equation}
        \label{e:Lyapprop}
            V_{i}(\gamma_i(t))\leq V_{i}(\gamma_i(0))\exp{(-\lambda_{i}t)}\quad\text{for all}\quad \gamma_i(0)\in\R^{d},\quad t\in[0,+\infty[,
        \end{equation}
        and $\gamma_i(\cdot)$ solves the $i$-th system dynamics in \eqref{e:family}, $i\in\P$.
        \end{fact}

        See Appendix \ref{ap1} for a short proof.

        The functions defined in \eqref{e:Lyaplike} and satisfying \eqref{e:Lyapprop} will be called Lyapunov-like functions in the sequel. We observe that for all $i,j\in\P$, the respective Lyapunov-like functions are related as follows: there exists $\mu_{ij}>0$ such that
        \begin{equation}
        \label{e:muijprop}
            V_{j}(\xi)\leq\mu_{ij}V_{i}(\xi), \quad\text{for all}\quad \xi\in\R^{d}
        \end{equation}
        whenever the switching from $i$ to $j$ is admissible. Indeed, we have the following \emph{tight} estimate of the numbers $\mu_{ij}$:

        \begin{prop}
        \label{p:muijestimate}
            Let the Lyapunov-like functions be defined as in \eqref{e:Lyaplike} with each $P_i$ symmetric and positive definite, $i\in\P$. Then the smallest constant $\mu_{ij}$ in \eqref{e:muijprop} is given by
        \begin{equation}
        \label{e:muijestimate}
            \mu_{ij} = \lambda_{\max}(P_{j}{P_{i}}^{-1}), \qquad i,j\in\P,
        \end{equation}
        where $\lambda_{\max}$ denotes the maximal eigenvalue.
        \end{prop}
        See Appendix \ref{ap1} for a proof.

        The assumption of linearly comparable Lyapunov functions, i.e., there exists $\mu\geq1$ such that
        \begin{equation}
        \label{e:muprop}
            V_{j}(\xi)\leq\mu V_{i}(\xi)\quad\text{for all}\quad \xi\in\R^{d}\quad\text{and all}\quad i,j\in\P,
        \end{equation}
        is standard in the theory of stability under average dwell time switching \cite[Theorem 3.2]{Liberzon}; \eqref{e:muprop} is a special case of \eqref{e:muijprop}.

		\subsection{Properties of switching signals}

        Let $h:[0,+\infty[\:\lra[0,+\infty[$ be a continuous monotone (strictly) increasing function, with $h(0) = 0$, and $\displaystyle{\lim_{r\to+\infty}h(r) = +\infty}$. To wit, $h$ is a class $\mathcal K_{\infty}$ function \cite[Def.\ 2.5]{Hahn_stability}.

        Recall from \S\ref{s:prelims} that $\Ntsigma$ denotes the number of switches on $]0,t]$ for $t>0$. We let
        \begin{equation}
        \label{e:nuh}
            \nu_{h}(t) \Let \nuh, \qquad t > 0,
        \end{equation}
        denote the \emph{h-frequency of switching} at $t$, and we define the \emph{asymptotic upper density} of $\nu_{h}$ as
        \begin{equation}
        \label{e:nuhat}
            \hat{\nu}_{h}\Let \ls_{t\to+\infty}\nu_{h}(t).
        \end{equation}
        We define the $i$-th holding time of a switching signal $\sigma$
        \begin{equation}
        \label{e:holdingtime}
            S_{i+1} \Let \diftau,\qquad i = 0,1,\cdots,
        \end{equation}
        where $\tau_{i}$ and $\tau_{i+1}$ denote two consecutive switching instants. For each pair $(k,\ell)\in E(\P)$, we define the transition frequency from vertex $k$ to vertex $\ell$ to be
        \begin{equation}
        \label{e:rho}
            \rho_{k\ell}(t) \Let \frac{\ktol}{\Ntsigma},\qquad t > 0,
        \end{equation}
        where $\ktol$ denotes the number of transitions from vertex $k$ to vertex $\ell$ till the last switching instant $\tau_{\Ntsigma}$ before (and including) $t$. Let
        \begin{equation}
        \label{e:rhohat}
            \hat{\rho}_{k\ell} \Let \ls_{t\to+\infty}\rho_{k\ell}(t)
        \end{equation}
        denote the \emph{upper asymptotic density} of $\rho_{k\ell}$. For each $j\in \P$, we define the \emph{h-fraction of activation} of system $j$ till time $t$ as
        \begin{equation}
        \label{e:eta}
            \eta_{h}(j,t) \Let \sum_{i:\sigma(\tau_{i})=j}\frac{\holdt}{h(t)},\qquad t > 0,
        \end{equation}
        and let
        \begin{align}
        \label{e:etahatcheck}
            \hat{\eta}_{h}(j) \Let \ls_{t\to+\infty}\eta_{h}(j,t)	\quad\text{and}\quad	\check{\eta}_{h}(j)\Let \li_{t\to+\infty}\eta_{h}(j,t)
        \end{align}
        denote the \emph{upper} and \emph{lower asymptotic densities} of $\eta_{h}(j)$, respectively.

    \section{Main result}
	\label{s:mainres}
        Our main result is the following:
        \begin{theorem}
		\label{t:main}
            Consider the family of systems \eqref{e:family}. Let $\P_{AS}$ and $\P_{U}\subset\P$ denote the sets of indices of asymptotically stable and unstable systems in \eqref{e:family}, respectively. For $i, j\in\P$ let the constants \(\lambda_i\) and \(\mu_{ij}\) be as in Fact \ref{fact:key} and \eqref{e:muijprop}, respectively. Then the switched system \eqref{e:swsys} is globally asymptotically stable for every switching signal $\sigma$ satisfying
            \begin{align}
				\label{e:swfreq} \check{\nu}_{h}\Let\li_{t\to+\infty}\nu_{h}(t)	& > 0\qquad \intertext{and}
            	\label{e:thmcondn} \hat{\nu}_{h}\sum_{(k,\ell)\in E(\P)}\hat{\rho}_{k\ell}\cdot\ln{\mu_{k\ell}}	& < \sum_{j\in \P_{AS}}\abs{\lambda_{j}}\check{\eta}_{h}(j) - \sum_{j\in \P_{U}}\abs{\lambda_{j}}\hat{\eta}_{h}(j),
            \end{align}
            where $\hat{\nu}_{h}$, $\hat{\rho}_{k\ell}$, $\hat{\eta}_{h}(j)$, and $\check{\eta}_{h}(j)$ are as in \eqref{e:nuhat}, \eqref{e:rhohat}, and \eqref{e:etahatcheck}, respectively.
        \end{theorem}
        A proof of Theorem \ref{t:main} is provided in Appendix \ref{ap2}. We explain Theorem \ref{t:main} and place it in the context of preceding works in the remainder of this section:

        \begin{remark}
		\label{r:discrete vs continuous}
			The quantity \(\hat\rho_{k\ell}\) in \eqref{e:thmcondn}---the asymptotic density of the \emph{number} of transitions from \(k\) to \(\ell\)---is concerned purely with the discrete transitions; it is not influenced by the continuous-time properties of the switching signal. In contrast, \(\hat \nu_h\), \(\check\eta_h(j)\), and \(\hat\eta_h(j)\) depend on continuous-time properties of the switching signal.
		\end{remark}

		\begin{remark}
            The class $\mathcal{K}_{\infty}$ function $h$ with respect to which we measure the asymptotic densities in \eqref{e:nuhat} and \eqref{e:etahatcheck} provides substantial flexibility in terms of admissibility of switching signals. Indeed, switching signals with $\Ntsigma$ being of the order $O(t\ln t)$, $O(t^{p})$ for $p>0$, as $t\rightarrow+\infty$, can be captured here: simply take $h(t) = t\ln(1+t)$, $t^{p}$, respectively, in the two cases above. In other words, $\Ntsigma\leq k_{1}t^{\frac{3}{2}}+k_{2}t+k_{3}$ for $k_{1}, k_{2}, k_{3}>0$, is permissible; in the last case Theorem \ref{t:main} can be applied with $h(t) = t^{\frac{3}{2}}$. It therefore enlarges, by a significant margin over preceding results, the class of switching signals under which global asymptotic stability of switched systems can be studied.
		\end{remark}

        \begin{remark}
        \label{r:theorem description}
            Observe that inequality \eqref{e:thmcondn} in Theorem \ref{t:main} involves \emph{only} asymptotic quantities related to the switching signal, namely, the asymptotic ``densities'' of the switching frequency, fraction of activation of the constituent systems and admissible transitions among them. On the left-hand side of \eqref{e:thmcondn}, the factor $\displaystyle{\sum_{(k,\ell)\in E(\P)}\hat{\rho}_{k\ell}\ln\mu_{k\ell}}$ multiplying the upper asymptotic density of the $h$-frequency $\nu_{h}$ of switching---a zeroth order property of the switching signal---contains the upper asymptotic density of ${\rho}_{k\ell}$, the frequency of admissible transitions among the systems in the given family \eqref{e:family}---a second order property of the switching signal. The terms on the right-hand side of \eqref{e:thmcondn} involve the switching destinations---a first order property of the switching signal: the first (\emph{resp}.\ second) term comprises of the lower (\emph{resp}.\ upper) asymptotic density of the $h$-fraction of activation of the asymptotically stable (\emph{resp}.\ unstable) systems in \eqref{e:family}, weighted by the corresponding quantitative measures of (in)stability.
        \end{remark}

        \begin{remark}
			The hypothesis \eqref{e:swfreq} $\displaystyle{\liminf_{t\to+\infty} \nu_{h}(t) > 0}$ is necessary to prevent the switched system from eventually ``adhering to'' an unstable system. This assumption is peculiar to our setting because we admit unstable systems in the family \eqref{e:family} in contrast to preceding works on average dwell time switching, where no unstable system in \eqref{e:family} is permitted.
		\end{remark}

		\begin{remark}
        \label{r:tightness}
            Although the condition \eqref{e:thmcondn} relies on the choice of the Lyapunov-like functions \(V_i\), we contend that in certain cases the condition \eqref{e:thmcondn} is \emph{tight}. Let us see what may happen at the ``boundary" of the inequality \eqref{e:thmcondn}. Consider, for instance, the family of scalar linear systems consisting of two elements:
            \[
                \dot x(t) = \lambda_i x(t),\qquad x(0) = x_0 > 0 \text{ given},\quad t\geq 0,
            \]
            where \(i \in \P \Let \{1, 2\}\), \(\lambda_1 = -\lambda_2 = \lambda\), where \(\lambda > 0\) is a fixed constant. Suppose the switching signal \(\sigma\) begins from \(i = 1\), is such that \(\tau_i = i\) for each \(i = 0, 1, \ldots\), and at each \(\tau_{i+1}\) we have \(\sigma(\tau_{i+1}) = \P\setminus\{\sigma(\tau_i)\}\). Straightforward calculations show that with \(h\) equal to the identity function, \(\check \nu_h = \hat \nu_h = 1\), and \(\check \eta_h(j) = \hat \eta_h(j) = \frac{1}{2}\) for each \(j\in\P\), and \(\hat\rho_{12} = \hat\rho_{21} = \frac{1}{2}\). With the obvious choice of Lyapunov-like functions $\R\ni z\longmapsto V_i(z) \Let \abs{z}^2$, one sees that \(\mu_{12} = \mu_{21} = 1\), which implies that the left-hand side of \eqref{e:thmcondn} is \(0\). In view of the data above, the right-hand side also evaluates to \(0\), which shows that \eqref{e:thmcondn} does not hold. Indeed, in this case the left-hand side is equal to the right-hand side instead of being strictly smaller; consequently, Theorem \ref{t:main} does not apply. Observe that the switched linear system
            \[
                \dot x(t) = \lambda_{\sigma(t)} x(t),\qquad x(0) = x_0 > 0 \text{ given}, \quad t\geq 0,
            \]
            for the given \(\sigma\) is \emph{not} GAS---the trajectories starting from non-zero initial conditions oscillate without approaching \(0\).
        \end{remark}

        \begin{remark}
        \label{r:adt}
            Recall \cite[\S3.2.2]{Liberzon} that a switching signal $\sigma$ is said to have average dwell time $\tau_{a}>0$ if there exist two positive numbers $N_{0}$ and $\tau_{a}$ such that $N_\sigma(T,t)\leq N_{0}+\frac{T-t}{\tau_{a}}$ for all $T \geq t\geq 0$, where $N_\sigma(T,t)$ denotes the number of discontinuities of a switching signal $\sigma$ on an interval $]t,T]$. (The case of \(N_0 = 1\) corresponds to switching signals having a fixed dwell time.) From the preceding inequality it is clear that the average dwell time condition imposes point-wise bounds on the number of switches---on \emph{any} time interval the number of switches can grow at most as an affine function of the length of the interval. In contrast, the condition \eqref{e:thmcondn} does not involve such point-wise bounds. The number of switches $\Ntsigma$, on the interval \(]0, t]\), can grow faster than an affine function of $t$; indeed, $\Ntsigma$ obeying $k_{0}t - k_0'\sqrt{t} \leq\Ntsigma \le k_{1}+k'_{1}t+k''_{1}\sqrt{t}$ for positive constants $k_{0}$, $k'_{0}$, $k_{1}$, $k'_{1}$, \(k''_1\), is perfectly admissible, with \(h\) being the identity function.
        \end{remark}

        \begin{remark}
        \label{r:uniformity}
            According to Theorem \ref{t:main}, under all switching signals $\sigma$ satisfying \eqref{e:thmcondn}, the switched linear system \eqref{e:swsys} is globally asymptotically stable. Recall \cite[$\S$2.1.1]{Liberzon} that the switched system \eqref{e:swsys} is globally uniformly asymptotically stable if there exists a class-$\mathcal{KL}$ function $\beta$ such that for all switching signals $\sigma$ and all initial conditions, the solutions of $\eqref{e:swsys}$ satisfy $\norm{x(t)} \leq \beta(\norm{x_0}, t)$ for all $t\geq0$. Uniform global asymptotic stability also arises in the context of restricted switching. Indeed, the average dwell time condition guarantees global uniform asymptotic stability over all switching signals with given average dwell time $\tau_{a} > \frac{\ln\mu}{2\lambda_0}$ and given ``chatter-bound'' $N_{0}$, where $\mu$ is as in \eqref{e:muprop}, and $N_{0}, \lambda_{0}$ are positive numbers as in \cite[Theorem 3.2, Remark 3.2]{Liberzon}. Theorem \ref{t:main}, however, does not guarantee uniform stability in this sense. To wit, let \(\sigma, \sigma'\) be two admissible switching signals satisfying the hypotheses of Theorem \ref{t:main}, and let \((x_\sigma(t))_{t\geq 0}\) and \((x_{\sigma'}(t))_{t\geq 0}\) be the corresponding solutions to \eqref{e:swsys}, respectively. In this setting,
			\begin{itemize}[label=\(\circ\), leftmargin=*]
				\item Theorem \ref{t:main} asserts that there are two class-$\mathcal{KL}$ functions \(\beta, \beta'\) such that, for all \(t\geq 0\) and all \(x_0\in\R^d\), \(\norm{x_\sigma(t)} \le \beta(\norm{x_0}, t)\) and \(\norm{x_{\sigma'}(t)} \le \beta'(\norm{x_0}, t)\);
				\item Theorem \ref{t:main} does \emph{not} claim that \(\beta = \beta'\).
			\end{itemize}
			This apparent deficiency is only natural in view of the fact that we place no restriction on the transient behavior of the switching signal; the latter would be crucial in obtaining uniform estimates.
        \end{remark}

		\begin{remark}
            The classical result \cite[Theorem 3.2]{Liberzon} on average dwell time switching, including the uniform stability assertions as mentioned in Remark \ref{r:uniformity}, follow immediately from straightforward and quite obvious modifications to our proof of Theorem \ref{t:main}.
        \end{remark}

        \begin{remark}
            Theorem \ref{t:main} strongly hints at a high degree of robustness inherent in stability of switched linear systems \eqref{e:swsys}. Indeed, note that only asymptotic properties of the switching signals matter insofar as GAS is concerned, despite the fact that there may be unstable systems in the family \eqref{e:family}. In addition, perturbations to $\sigma$ that decay faster than the dominant rate of switching, as measured by the class $\mathcal{K}_\infty$ function $h$, do not destabilize a GAS switched linear system. This feature is currently under investigation, and will be reported elsewhere.
        \end{remark}

        \begin{remark}
             Theorem \ref{t:main} extends readily to the case of discrete-time switched linear systems, and also to the case of switched nonlinear systems under appropriate and standard assumptions. This theme will be expanded upon and reported separately.
        \end{remark}

    \section{Numerical example}
	\label{s:example}


        In this section we present a simple numerical example.

        \subsection{The systems}
         We consider $\P = \{1, 2, 3, 4\}$ with $\P_{AS} = \{1\}$, $\P_{MS} = \{2\}$, and $\P_{U} = \{3,4\}$, with
         \begin{align*}
         A_{1} &= \pmat{-0.1 & -0.2\\0.1 & -0.4},& A_{2} &= \pmat{0 & 0.1\\-0.1 & 0},\\
         A_{3} &= \pmat{0.1 & 0.2\\0.4 & 0.3},& A_{4} &= \pmat{0.2 & 0.1\\0.3 & 0}.
         \end{align*}
         The pairs $(P_{i},\lambda_{i})$ discussed in Fact \ref{fact:key}, are obtained as:
         \begin{align*}
            (P_{1},\lambda_{1}) &= \Biggl(\pmat{3.8333 & -1.1667\\-1.1667 & 1.8333},0.2288\Biggr),&
            (P_{2},\lambda_{2}) &= \Biggl(\pmat{1 & 0\\0 & 1},0\Biggr),\\
            (P_{3},\lambda_{3}) &= \Biggl(\pmat{1 & 0\\0 & 1}, -1.0797\Biggr),&
            (P_{4},\lambda_{4}) &= \Biggl(\pmat{1 & 0\\0 & 1}, -0.7301\Biggr).
         \end{align*}
        A directed graph representation of the switching signal is below.

        \begin{center}
        \begin{tikzpicture}[every path/.style={>=latex},every node/.style={draw,circle}]
            \node            (a) at (0,0)  { 1 };
            \node            (b) at (2,0)  { 3 };
            \node            (c) at (2,-2) { 2 };
            \node            (d) at (0,-2) { 4 };

            \draw[->] (a) edge (b);
            \draw[<-] (a) edge [bend left] (b);
            \draw[->] (a) edge (c);
            \draw[<-] (a) edge [bend left] (c);
            \draw[->] (a) edge (d);
            \draw[<-] (a) edge [bend right] (d);
            \draw[->] (b) edge (c);
            \draw[<-] (b) edge [bend left] (c);
            \draw[->] (c) edge (d);
            \draw[<-] (c) edge [bend left] (d);

        \end{tikzpicture}
        \end{center}

        \subsection{The switching signal}
        Let $\varepsilon = 0.1$ and $\eta = 0.5$. Let $h$ be the identity function, and suppose that $\sigma$ is such that
        \begin{itemize}[label=\(\circ\), leftmargin=*]
            \item $\Ntsigma = \frac{t}{3}+t^{1-\eta}$, $\eta_{h}(1,t) = \frac{t}{1+\eps}-t^{1-\eta}$, $\eta_{h}(2,t) = t^{\frac{1}{9}}$, $\eta_{h}(3,t) = \eta_{h}(4,t) = \frac{1}{2}(t-(\eta_{h}(1,t)+\eta_{h}(2,t)))$, and
            \item ${\rho}_{k\ell} = \frac{1}{10}$ for each pair $(k,\ell)\in E(\P)$.
        \end{itemize}
        Consequently,
        \begin{itemize}[label=\(\circ\), leftmargin=*]
            \item $\hat{\eta}_{h}(3) = \hat{\eta}_{h}(4) = \frac{\varepsilon}{2(1+\varepsilon)}$ and $\check{\eta}_{h}(1) = \frac{1}{1+\varepsilon}$; note that systems 3 and 4 are unstable while 1 is asymptotically stable,
            \item $\hat{\nu}_{h} = \frac{1}{3}$, and
            \item $\hat{\rho}_{k\ell} = \frac{1}{10}$ for each pair $(k,\ell)\in E(\P)$.
        \end{itemize}
        We calculate $\mu_{k\ell}$ for each pair $(k,\ell)\in E(\P)$ using \eqref{e:muijestimate}, and obtain
        \begin{align*}
        \mu_{12} &= 0.7712 , & \mu_{13} &= 0.7712 , & \mu_{14} &= 0.7712 , & \mu_{21} &= 4.3699 , & \mu_{23} &= 1 , &\\
        \mu_{24} &= 1, & \mu_{31} &= 4.3699, & \mu_{32} &= 1, & \mu_{41} &= 4.3699, & \mu_{42} &= 1.
        \end{align*}
        Observe from the above directed graph that the total number of admissible transitions is $10$. We verify that
            \[
                \displaystyle{\hat{\nu}_{h}\sum_{(k,\ell)\in E(\P)}\hat{\rho}_{k\ell}\ln\mu_{k\ell} = 0.12149}
            \]
            and
            \[
                \displaystyle{\sum_{j\in\P_{AS}}\abs{\lambda_{j}}\check{\eta}_{h}-\sum_{j\in\P_{U}}\abs{\lambda_{j}}\hat{\eta}_{h} = 0.12574}.
            \]
        Therefore, \eqref{e:thmcondn} holds, and Theorem \ref{t:main} asserts that the switched system \eqref{e:swsys} with the above data is GAS.

        \subsection{The verification}

        We study the response for $(x(t))_{t\geq0}$ corresponding to different initial conditions $x_{0}$ and observe that in each case $x(t)\to 0$ as $t\to+\infty$. The responses for $x_{0} =  \pmat{-1000\\1000}$ and $x_{0} =  \pmat{500\\-1200}$ are given in Figures \ref{f:response1} and \ref{f:response2}, respectively. Observe the sluggish rate of convergence in each case.

        \begin{figure}[htbp]
        \begin{center}
            \includegraphics[scale=0.5]{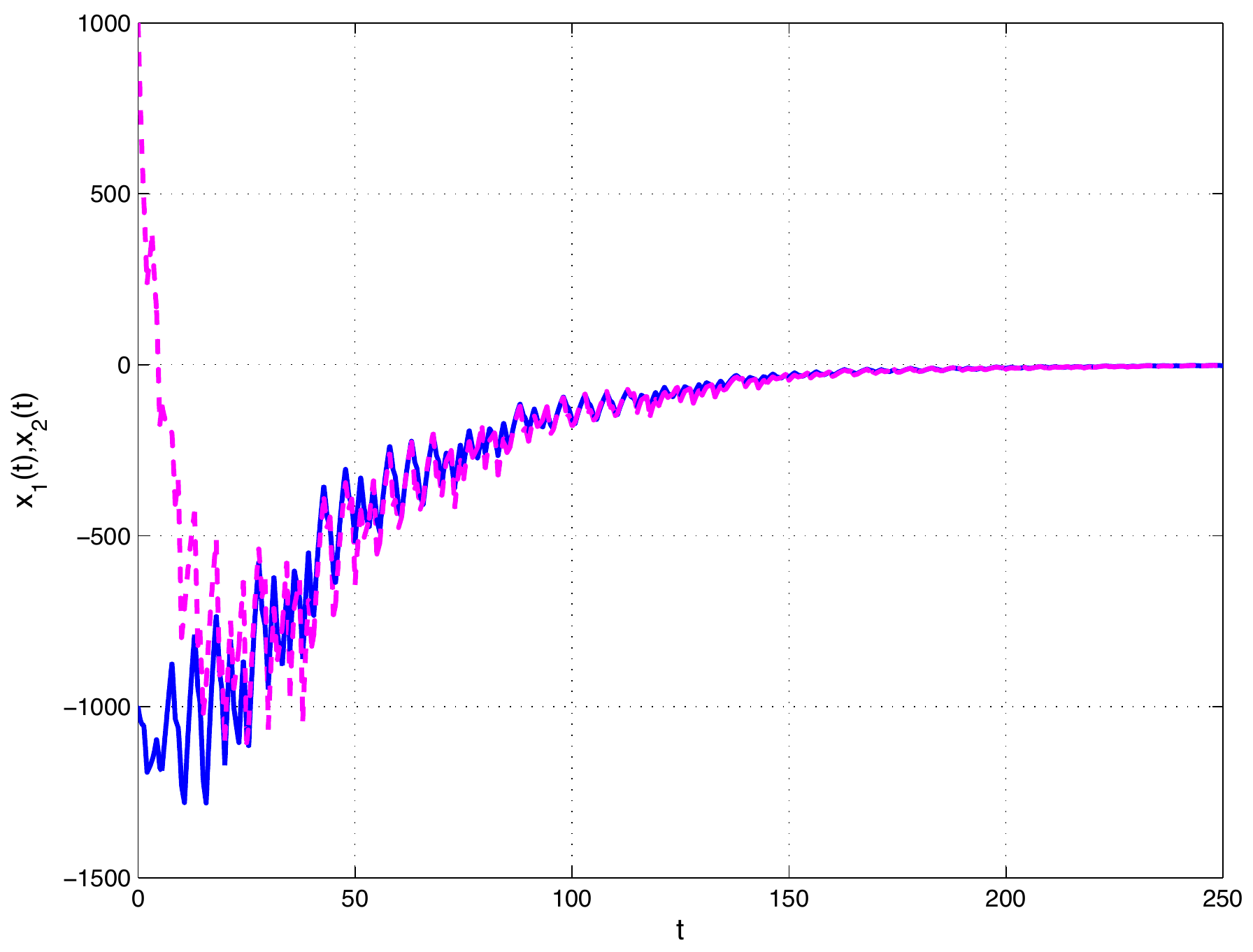}
            \caption{Solution to \eqref{e:swsys} with $x_{1}(0) = -1000, x_{2}(0) = 1000$. The graphs of $x_{1}(\cdot)$ and $x_{2}(\cdot)$ are plotted in solid and dashed lines, respectively}
            \label{f:response1}
        \end{center}
        \end{figure}

        \begin{figure}[htbp]
        \begin{center}
            \includegraphics[scale=0.5]{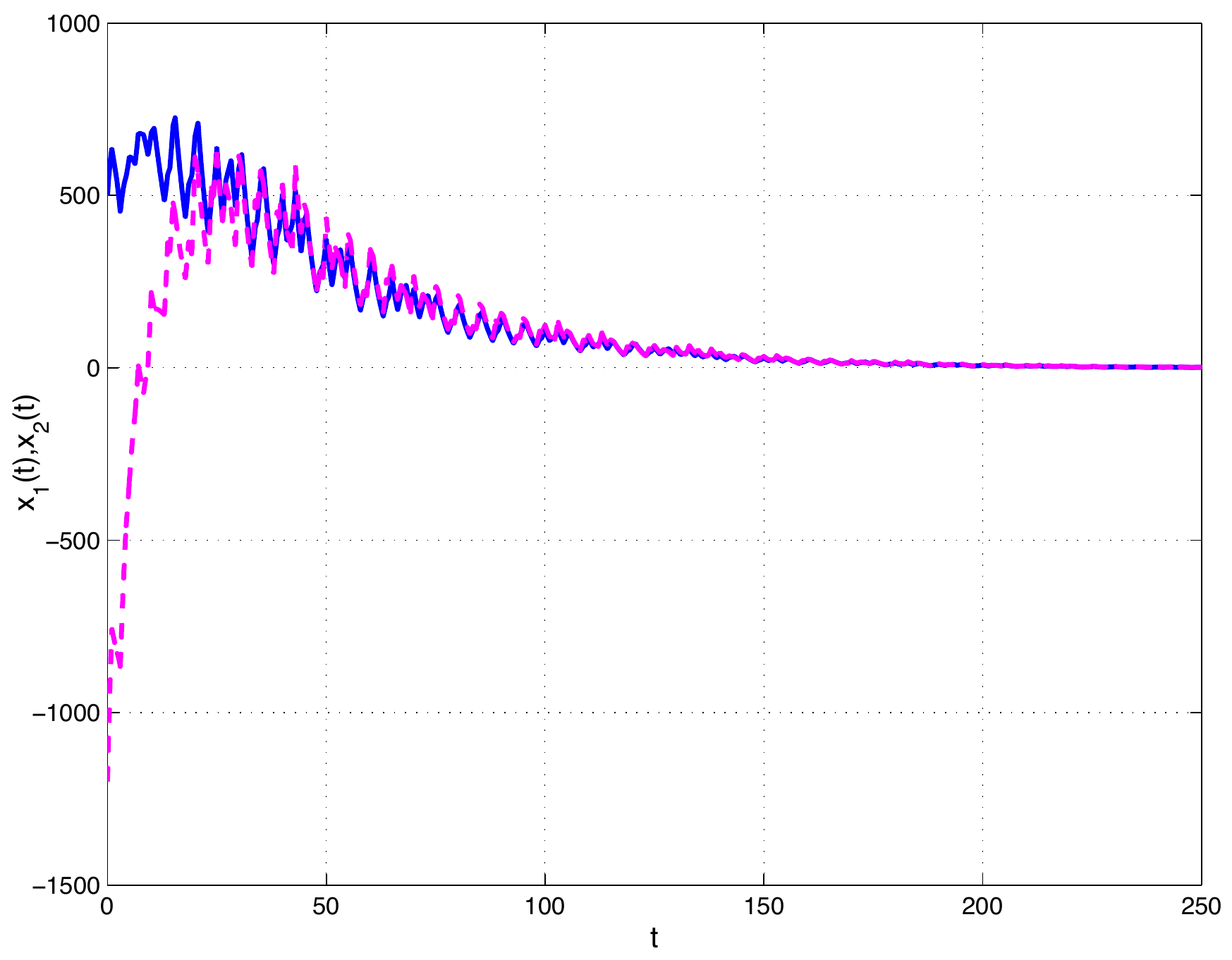}
            \caption{Solution to \eqref{e:swsys} with $x_{1}(0) = 500, x_{2}(0) = -1200$. The graphs of $x_{1}(\cdot)$ and $x_{2}(\cdot)$ are plotted in solid and dashed lines, respectively}
            \label{f:response2}
        \end{center}
        \end{figure}

    \section{Concluding remarks}
	\label{s:conclusion}
        In this article we introduced a class of switching signals under which global asymptotic stability of a continuous-time switched linear system is guaranteed. We considered unstable systems in the family and involved only asymptotic behaviour of the zeroth order (frequency of switching), first order (switching destinations), and second order (admissible transitions among these destinations) properties of the switching signal to characterize the proposed class of switching signals. Further work in this direction will focus on aspects of reachability $\grave{a}$ la \cite{Liberzon_reach} of switched systems under such switching signals.

    \appendix
    \section{}
    \label{ap1}
    This appendix collects our proofs of the facts presented in \S\ref{s:prelims}. Throughout this appendix $\lambda_{\max}(M)$ and $\lambda_{\min}(M)$ denote the maximal and minimal eigenvalues of a square matrix $M$.

    \begin{proof}[Proof of Fact \ref{fact:gas}]
        Recall that by definition, uniform global asymptotic convergence of \eqref{e:swsys} is equivalent to the following property:
		\begin{quote}
			for all $\eps>0$ and for all $r>0$ there exists $T(r,\eps)>0$ such that $\norm{x_{0}}<r$ implies $\norm{x(t)}<\eps$ for every $t>T(r,\eps)$,
		\end{quote}
		and Lyapunov stability of \eqref{e:swsys} is equivalent to:
		\begin{quote}
			for all $\eps>0$ there exists $\delta(\eps)>0$ such that $\norm{x_{0}}<\delta(\eps)$ implies $\norm{x(t)}<\eps$.
		\end{quote}
		We need to establish Lyapunov stability of \eqref{e:swsys} under any switching signal \(\sigma\) that ensures that \eqref{e:swsys} is globally asymptotically convergent.\par
		\noindent\textsf{Step 1.} Let \(\eps > 0\) be given and \(\sigma\) be an admissible switching signal such that \eqref{e:swsys} enjoys the global asymptotic convergence property. To wit, there exists \(T(1, \eps) > 0\) such that \(\norm{x(t)} < \eps\) for all \(t > T(1, \eps)\) whenever \(\norm{x_0} < 1\).\par
		\noindent\textsf{Step 2.} Note that the family \eqref{e:family} is globally uniformly Lipschitz, with \(L \Let \max_{i\in\P} \norm{A_i}\) serving as the uniform Lipschitz constant over \(\P\). It follows that if \((x(t))_{t\geq 0}\) is the solution to \eqref{e:swsys} under \emph{any} admissible switching signal \(\sigma\), then
		\[
			\abs{\frac{\drv}{\drv t}\norm{x(t)}^2} = \abs{2 x(t)\transp A_{\sigma(t)} x(t)} \le 2 L \norm{x(t)}^2
		\]
		by the Cauchy-Schwarz inequality. It follows at once that
		\[
			\norm{x_0} \exp(-L t) \le \norm{x(t)} \le \norm{x_0} \exp(Lt)\qquad\text{for all }t\geq 0.
		\]
		Selecting \(\delta' = \eps \exp(-LT(1, \eps))\), we see from the preceding inequality that \(\norm{x(t)} < \eps\) for all \(t \in[0, T(1, \eps)]\) whenever \(\norm{x_0} < \delta'\) and \(\sigma\) is arbitrary (but admissible).\par
		\noindent\textsf{Step 3.} It remains to specialize to the switching signal \(\sigma\) fixed in \textsf{Step 1}, and select \(\delta = \min\{1, \delta'\}\) to verify Lyapunov stability of \eqref{e:swsys}.
    \end{proof}

    \begin{proof}[Proof of Fact \ref{fact:key}]
        For asymptotically stable systems let $\R^d\ni z\longmapsto V_{i}(z) \Let z^{\top}P_{i}z$, where $P_i \in\R^{d\times d}$ is the symmetric positive definite solution to the Lyapunov equation
        \begin{equation}
		\label{e:Lyap equation}
            {A_{i}}^{\top}P_{i}+P_{i}A_{i}+Q_{i} = 0
        \end{equation}
		for some pre-selected symmetric and positive definite matrix \(Q_i\in\R^{d\times d}\) \cite[Corollary 11.9.1]{Bernstein}. If \(A_i\) is marginally stable, it is known \cite[Corollary 11.9.1]{Bernstein} that there exists a symmetric and positive definite matrix \(P_i\in\R^{d\times d}\) and a symmetric and non-negative definite matrix \(Q_i\in\R^{d\times d}\) that solve the Lyapunov equation \eqref{e:Lyap equation}; we put \(\R^d\ni z\longmapsto V_i(z) \Let z^\top P_i z\) as the corresponding Lyapunov-like function. A straightforward calculation gives
        \[
            \frac{\drv}{\drv t}V_{i}(\gamma_i(t)) = \inprod{\nabla V_{i}(\gamma_i (t))}{A_{i}\gamma_i(t)} = -\gamma_i(t)^\top Q_{i}\gamma_i(t)
        \]
		in both cases, where \(\gamma_i(\cdot)\) solves the \(i\)-th system dynamics in \eqref{e:family}. Since for any symmetric matrix \(\Xi\in\R^{d\times d}\) we have (\cite[Lemma 8.4.3]{Bernstein})
        \[
            \lambda_{\min}(\Xi)\norm{z}^2\leq z^{\top}\Xi z\leq\lambda_{\max}(\Xi)\norm{z}^2 \quad\text{for all}\quad z\in\R^{d},
        \]
        for all $z\in\R^{d}$ we get
        \[
            -z^{\top}Q_{i}z\leq-\frac{\lambda_{\min}(Q_{i})}{\lambda_{\max}(P_{i})}z^{\top}P_{i}z.
        \]
        Defining $\lambda_{i} = \frac{\lambda_{\min}(Q_{i})}{\lambda_{\max}(P_{i})}$, we arrive at
        \[
            \frac{\drv}{\drv t}V_{i}(\gamma_i(t))\leq -\lambda_{i}V_{i}(\gamma_i(t))
        \]
        which gives \eqref{e:Lyapprop} with $\lambda_{i}\geq0$.

        For unstable systems, let us consider the simplest case of a symmetric and positive definite matrix $P_{i} = I_{d}$, and let $\R^{d}\ni z\longmapsto V_{i}(z) \Let \norm{z}^{2}$. Then, by the Cauchy-Schwarz inequality,
    \begin{align*}
        \abs{\frac{\drv}{\drv t}V_{i}(\gamma_i(t))} &= \abs{\inprod{\nabla V_{i}(\gamma_i(t))}{A_{i}\gamma_i(t)}}\\
        & \le 2\norm{A_{i}\gamma_i(t)}\norm{\gamma_i(t)}\\
        & \le 2\norm{A_{i}}V_{i}(\gamma_i(t)).
    \end{align*}
    To wit,
    \[
    	-2\norm{A_{i}}V_{i}(\gamma_i(t))\leq\frac{\drv}{\drv t}V_{i}(\gamma_i(t))\leq 2\norm{A_{i}}V_{i}(\gamma_i(t)),
    \]
    which implies that
    \[
        V_{i}(\gamma_i(t))\leq V_{i}(\gamma_i(0))\exp(2\norm{A_{i}}t)\quad\text{for all}\quad t\geq0.
    \]
    Therefore, \eqref{e:Lyapprop} holds for unstable systems with $\lambda_{i}$ negative.
    \end{proof}

    \begin{proof}[Proof of Proposition \ref{p:muijestimate}]
    	Since $P_{i}$ is symmetric and positive definite, it is non-singular. Observe that $P_{j}{P_{i}}^{-1}$ is similar to $P_{i}^{-1/2}(P_{j}P_{i}^{-1})P_{i}^{1/2}$, and that the matrix $P_{i}^{-1/2}P_{j}P_{i}^{-1/2}$ is symmetric and positive definite. Since the spectrum of a matrix is invariant under similarity transformations, the eigenvalues of $P_{j}{P_{i}}^{-1}$ are the same as the eigenvalues of $P_{i}^{-1/2}P_{j}P_{i}^{-1/2}$; consequently, the eigenvalues of $P_{j}{P_{i}}^{-1}$ are real numbers. In addition,
    	\begin{align*}
       	 \sup_{0\neq z\in\R^{d}}\frac{\inprod{P_{j}z}{z\rangle}}{\inprod{P_{i}z}{z}} &= \sup_{0\neq z\in\R^{d}}\frac{\inprod{P_{j}z}{z}}{\inprod{P_{i}^{1/2}z}{P_{i}^{1/2}z}}&&\\
       	 &= \sup_{0\neq y\in\R^{d}}\frac{\inprod{P_{j}(P_{i}^{-1/2}y)}{P_{i}^{-1/2}y}}{\inprod{y}{y}}&&\text{with $z \Let P_{i}^{-1/2}y$}&&\\
       	 &= \sup_{0\neq y\in\R^{d}}\frac{\inprod{P_{i}^{-1/2}P_{j}P_{i}^{-1/2}y}{y}}{\langle y,y\rangle}&&\\
       	 &= \lambda_{\max}(P_{i}^{-1/2}P_{j}P_{i}^{-1/2})&&\\
       	 &= \lambda_{\max}(P_{j}P_{i}^{-1}),
    	\end{align*}
    	Since ${V_{j}(z)}\leq\mu_{ij}{V_{i}(z)}$ for all $z\in\R^{d}$, the smallest constant $\mu_{ij}$ satisfies \eqref{e:muijestimate}. A similar computation appears in \cite[Lemma 13]{Aneel11}.
    \end{proof}

    \section{}
    \label{ap2}
    This appendix contains our

    \begin{proof}[Proof of Theorem \ref{t:main}]
    We will employ properties of the quadratic Lyapunov-like function $V_{i}$ for all $i\in\P$. Recall that $\swinstants$ are the switching instants before (and including) $t>0$.

    In view of \eqref{e:Lyapprop},
    \begin{align}
    \label{e:proof1}
        V_{\sigma(t)}(x(t))\leq\exp{(-\lambda_{\sigma(\tau_{\Ntsigma})}(t-\tau_{\Ntsigma}))}V_{\sigma(t)}(x(\tau_{\Ntsigma})).
    \end{align}
    By \eqref{e:muijprop},
    \begin{align*}
        V_{\sigma(\tau_{i+1})}(x(\tau_{i+1}))\leq\mu_{\sigma(\tau_{i})\sigma(\tau_{i+1})}V_{\sigma(\tau_{i})}(x(\tau_{i+1})).
    \end{align*}
    In conjunction with \eqref{e:proof1}, a straightforward iteration leads to
    \begin{equation}
    \label{e:proof2}
	\begin{aligned}
        V_{\sigma(t)}(x(t))	&\leq V_{\sigma_{(0)}}(x_{0})\muprod\\
        					&\quad\exp\Bigl(-\lambda_{\sigma(\tau_{\Ntsigma})}(\dift)-\tausum\Bigr).
	\end{aligned}
    \end{equation}

    First of all, defining $\bar{\mu}_{ij} = \ln{\mu_{ij}}$, we have
    \begin{align}
    \muprod &= \exp\Biggl(\sum_{i=0}^{\Ntsigma}\bar{\mu}_{\sigma(\tau_{i})\sigma(\tau_{i+1})}\Biggr) = \exp\left(\sum_{k\in \P}\sum_{i=0}^{\Ntsigma}\sum_{\substack{k\rightarrow\ell:\\l\in \P,\\k\neq\ell,\\\sigma(\tau_{i})=k,\\\sigma(\tau_{i+1})=l}}\bar{\mu}_{kl}\right)\nonumber\\
    &= \exp\biggl(\sum_{(k,l)\in E(\P)}\bar{\mu}_{kl}\cdot\ktol\biggr)\nonumber\\
    &= \exp\biggl(\Ntsigma\sum_{(k,l)\in E(\P)}\bar{\mu}_{kl}\cdot\frac{\ktol}{\Ntsigma}\biggr) \label{e:proof3},
    \end{align}
    Recall that $\ktol$ denotes the number of transitions from vertex $k$ to vertex $\ell$ till the last switching instant $\tau_{\Ntsigma}$ before (and including) $t$. Employing the definition of holding times in \eqref{e:holdingtime},
    \begin{align*}
        \exp\biggl(-\tausum\biggr) &=
        \exp\biggl(-\sum_{i=0}^{\Ntsigma-1}\lambda_{\sigma(\tau_{i})}\holdt\biggr)\\ &= \exp\biggl(-\sum_{i=0}^{\Ntsigma-1}\biggl(\sum_{j\in \P}\indicator\lambda_{j}\holdt\biggr)\biggr).
    \end{align*}
    Separating out the asymptotically stable, marginally stable, and unstable systems in the family \eqref{e:family} into the subsets $\P_{AS}$, $\P_{MS}$, and $\P_{U}\subset\P$, respectively, we see that the argument of the exponential in the right-hand side above is equal to
    \[
        -\biggl(\astablesum+\mstablesum+\ustablesum\biggr).
    \]
    Recall that $\lambda_{j}$'s are positive, zero, and negative for the first, middle, and the last sums, respectively. Thus, the expression above equals
    \begin{equation}
    \label{e:proof4}
        \ustablesummag-\astablesummag.
    \end{equation}
    Substituting \eqref{e:proof3} and \eqref{e:proof4} into \eqref{e:proof2}, we arrive at
    \begin{equation}
	\label{e:key step}
        V_{\sigma(t)}(x(t)) \leq \exp\bigl(\psi(t)\bigr)V_{\sigma(0)}(x_{0}),
    \end{equation}
    where, for $t>0$, we define the function \(\psi\) as
    \begin{equation}
    \begin{split}
    \label{e:proof5}
        \psi(t) &\Let \Ntsigma\sum_{(k,l)\in E(\P)}\bar{\mu}_{kl}\frac{\ktol}{\Ntsigma}+\ustablesummag\\&\quad-\astablesummag-\lambda_{\sigma(\tau_{\Ntsigma})}(\dift).
    \end{split}
    \end{equation}

    Second of all, we see that the right-hand-side of \eqref{e:proof5}, for $t>0$, is equal to
    \begin{equation}
    \begin{split}
        h(t)\biggl(\nuh\sum_{(k,\ell)\in E(\P)}{\bar{\mu}_{k\ell}}\cdot\frac{\ktol}{\Ntsigma}+\ustablesummagfreq\\-\astablesummagfreq-\lambda_{\sigma(\tau_{\Ntsigma})}\frac{(\dift)}{h(t)}\biggr).
    \end{split}
    \end{equation}
    Recall that $\nu_{h}(t) \Let \nuh$ is the $h$-frequency of the switching signal $\sigma$ at $t>0$. Let us define two functions
    \[
        ]0,+\infty[\:\ni t\longmapsto f(t) \Let \sum_{(k,\ell)\in E(\P)}\bar{\mu}_{k\ell}\cdot\frac{\ktol}{\Ntsigma}
    \]
    and
    \begin{align*}
        ]0,+\infty[\:\ni t\longmapsto g(t) &\Let \ustablesummagfreq\\&-\astablesummagfreq-\lambda_{\sigma(\tau_{\Ntsigma})}\frac{(\dift)}{h(t)}.
    \end{align*}

    To verify GAS of the switched system \eqref{e:swsys}, (by Definition \ref{d:GAS},) we need to find conditions such that
    \begin{align}
    \label{e:proof7} 		& \lim_{t\rightarrow+\infty}\exp\Bigl(\bigl(\nu_{h}(t)f(t)+g(t)\bigr)h(t)\Bigr) = 0
		\intertext{and}
	\label{e:uniformity}	& \text{convergence is uniform for initial conditions \(\tilde x_0\) satisfying \(\norm{\tilde x_0} \le \norm{x_0}\)}.
    \end{align}
    Clearly, a sufficient condition for \eqref{e:proof7} is that
    \begin{equation}
    \label{e:proof8}
        \ls_{t\rightarrow+\infty}\bigl(\nu_h(t)f(t)+g(t)\bigr)<0,
    \end{equation}
    so our proof will be complete if we establish \eqref{e:proof8} and verify \eqref{e:uniformity} separately. The following steps are geared towards establishing \eqref{e:proof8}.

    Recall \cite[\S0.1]{LojaReal} that for two real valued functions $\vartheta$ and $\varphi$,
    \begin{align}
        \ls_{t\rightarrow+\infty}(\vartheta(t)+\varphi(t))\leq\ls_{t\rightarrow+\infty}\vartheta(t)+\ls_{t\rightarrow+\infty}\varphi(t)   \label{e:proof9}\\
        \li_{t\rightarrow+\infty}(\vartheta(t)+\varphi(t))\geq\li_{t\rightarrow+\infty}\vartheta(t)+\li_{t\rightarrow+\infty}\varphi(t)   \label{e:proof10}
    \end{align}
    provided the right-hand sides are not of the form $\mp\infty\pm\infty$.
    Similarly, if $\vartheta,\varphi\geq0$, then
    \begin{equation}
    \label{e:proof11}
        \ls_{t\rightarrow+\infty}(\vartheta(t)\cdot\varphi(t))\leq\ls_{t\rightarrow+\infty}\vartheta(t)\cdot\ls\varphi(t)
    \end{equation}
    provided the right-hand side is not of the form $0\cdot\pm\infty$. By definition of $\nu_{h}(t)$, $f(t)$ and $g(t)$, we see that properties \eqref{e:proof9}-\eqref{e:proof11} hold for the function $t\longmapsto\nu_{h}(t)f(t)+g(t)$.

    In view of \eqref{e:proof9} and \eqref{e:proof11},
    \begin{align*}
        \ls_{t\rightarrow+\infty}(\nu_{h}(t)f(t)+g(t))\leq\ls_{t\rightarrow+\infty}\nu_{h}(t)\ls_{t\rightarrow+\infty}f(t)+\ls_{t\rightarrow+\infty}g(t).
    \end{align*}
    By \eqref{e:proof9}, \eqref{e:rho}, and the definition of $f$, we get
    \begin{equation}
    \label{e:proof13}
        \ls_{t\rightarrow+\infty}f(t)\leq \sum_{(k,\ell)\in E(\P)}\bar{\mu}_{k\ell}\ls_{t\rightarrow+\infty}\rho_{k\ell}(t).    \nonumber
    \end{equation}
    The hypothesis \eqref{e:swfreq} guarantees that the term $-\lambda_{\sigma(\tau_{\Ntsigma})}\frac{(\dift)}{h(t)}$ in the expression of $g(t)$ is $o(h(t))$ as $t\to+\infty$. Indeed, $t-\tau_{\Ntsigma}\neq o(h(t))$ as $t\to+\infty$ implies that $\check{\nu}_{h} = 0$, which contradicts our assumption \eqref{e:swfreq}. In other words, $\frac{t-\tau_{\Ntsigma}}{h(t)}\to 0$ as $t\to+\infty$, i.e., $\frac{t-\tau_{\Ntsigma}}{h(t)} = o(1)$ as $t\to+\infty$. Hence, by definition \eqref{e:eta}, we have
    \begin{equation*}
        \ls_{t\rightarrow+\infty}g(t) = \ls_{t\rightarrow+\infty}\Bigl(\ustablesummagfreqnew-\astablesummagfreqnew\Bigr).
    \end{equation*}
    In view of \eqref{e:proof9}, the right-hand side is bounded above by
    \[
        \ls_{t\rightarrow+\infty}\ustablesummagfreqnew-\li_{t\rightarrow+\infty}\astablesummagfreqnew.
    \]
    Therefore,
    \begin{align*}
        \ls_{t\rightarrow+\infty}g(t)\leq\ustablesummagfreqls-\astablesummagfreqli.
    \end{align*}
    Thus for \eqref{e:proof7} to hold, it is sufficient that
    \begin{align}
    \label{e:proof15}
        \hat{\nu}_{h}\sum_{(k,\ell)\in E(\P)}\bar{\mu}_{k\ell}\cdot\ls_{t\rightarrow+\infty}{\rho}_{k\ell}(t)+\ustablesummagfreqls\nonumber\\-
        \astablesummagfreqli<0.
    \end{align}
    Employing the definitions in \eqref{e:rhohat}, \eqref{e:etahatcheck}, and noting that $\bar{\mu}_{k\ell} = \ln{\mu_{k\ell}}$, we see that \eqref{e:proof15} holds in view of \eqref{e:thmcondn}. Thus, \eqref{e:proof7} holds.
	
	It remains to verify \eqref{e:uniformity}. We provide a simple direct argument below; a more general result appears in \cite[Theorem 2]{Sontag08}. To this end, we get back to \eqref{e:key step}. With \([0, +\infty[\:\ni r\longmapsto \ol\alpha(r) \Let \lambda_{\max}\bigl(\sum_{i\in\P} P_i\bigr) r^2 \in[0, +\infty[\) and \([0, +\infty[\:\ni r\longmapsto \underline{\alpha}(r) \Let \min_{i\in\P}\lambda_{\min}(P_i) r^2\), we see that \(\underline{\alpha}(\norm{z}) \le V_{i}(z) \le \ol\alpha(\norm{z})\) for all \(i\in\P\) and \(z\in\R^d\). In conjunction with \eqref{e:key step}, we get
	\[
		\underline{\alpha}(\norm{x(t)}) \le V_{\sigma(t)}(x(t)) \le \ol\alpha(\norm{x_0}) \exp\bigl(\psi(t)\bigr)\quad\text{for all }t\ge 0,
	\]
	which implies, for \(c \Let \sqrt{\frac{\lambda_{\max}\bigl(\sum_{i\in\P} P_i\bigr)}{\min_{i\in\P}\lambda_{\min}(P_i)}}\),
	\begin{equation}
	\label{e:key uniformity step}
		\norm{x(t)} \le c \norm{x_0} \exp\bigl(\tfrac{1}{2}\psi(t)\bigr)\quad\text{for all }t\ge 0.
	\end{equation}
	Since the initial condition \(x_0\) is decoupled from \(\psi\) on the right-hand side of \eqref{e:key uniformity step}, if \(\norm{x(t)} < \eps\) for all \(t > T(\norm{x_0}, \eps)\) for some pre-assigned \(\eps > 0\), then the solution \((\tilde x(t))_{t\ge 0}\) to \eqref{e:swsys} corresponding to an initial condition \(\tilde x_0\) such that \(\norm{\tilde x_0} \le \norm{x_0}\) satisfies \(\norm{\tilde x(t)} < \eps\) for all \(t > T(\norm{x_0}, \eps)\). Uniform global asymptotic convergence follows at once, and therefore, so does the assertion of Theorem \ref{t:main}.
    \end{proof}

\end{document}